\documentclass{article}
\usepackage{color}
\usepackage{graphicx}
\usepackage{epstopdf}
\usepackage{amsmath}
\usepackage[justification=centering]{caption}
\usepackage{amssymb}
\usepackage{mathrsfs}
\usepackage{float}
\usepackage[numbers,sort&compress]{natbib}
\usepackage{amsthm}
\numberwithin{equation}{section}
\usepackage{pgf}
\usepackage{CJK}
\usepackage{graphicx}
\usepackage{tikz}
\usepackage{stmaryrd}
\usepackage{graphicx}
\usepackage[T1]{fontenc}
\usepackage[english]{babel}
\usepackage{listings}
\usepackage{xcolor,mathrsfs,url}
\usepackage{ifthen}
\newtheorem{theorem}{Theorem}

\newtheorem{proposition}{Proposition}

\usepackage{caption}
\usepackage{float}
\usepackage{subfigure}

\usepackage{hyperref}
\hypersetup{hypertex=true,
            colorlinks=true,
            linkcolor=blue,
            anchorcolor=blue,
            citecolor=blue}
\begin{document}
          \title{Painlev\'{e}-type  asymptotics  for  the defocusing  Hirota equation in transition  region}
\author{Weikang XUN$^1$, Luman JU$^2$ and Engui FAN$^{1}$\thanks{\ Corresponding author and email address: faneg@fudan.edu.cn } }
\footnotetext[1]{ \  School of Mathematical Sciences  and Key Laboratory of Mathematics for Nonlinear Science, Fudan University, Shanghai 200433, P.R. China.}

\date{ }

\maketitle

\begin{abstract}
	\baselineskip=17pt
      We consider   the Cauchy   problem for the classical Hirota equation on the line with decaying initial data. 
    Based on the spectral analysis of the Lax pair of the Hirota equation, 
 we  first expressed the solution of the Cauchy   problem in terms of the solution of a Riemann-Hilbert problem. 
 Further we apply nonlinear steepest descent analysis to  obtain the long-time  asymptotics   
 of the solution in the critical transition region $|\frac{x}{t} - \frac{\alpha^2}{3\beta}|t^{2/3}  \leq M$, $M$ is a positive constant. 
  Our  result shows  that the long time asymptotics of the Hirota equation can be expressed  in terms of the solution of   Painlev\'{e}  $\mathrm{II}$  equation.
\end{abstract}

\baselineskip=17pt

\newpage

\tableofcontents

\section{Introduction}

    In 1973, Hirota first derived the following equation \cite{Hirota}
\begin{equation}\label{hirota}
      i u_{t} + \alpha u_{xx} + i \beta u_{xxx}  + 3 i \gamma  |u|^2  u_{x}  + \delta |u|^2 u =0,
\end{equation}
where $u(x,t)$ is a complex-valued scalar function, $\alpha,\beta,\gamma$ and $\delta$ are real constants which satisfy $\alpha \gamma = \beta \delta$.
Especially for  $\alpha =  \frac{1}{2} \delta$ and $\beta = \frac{1}{2} \gamma$,   the equation  (\ref{hirota})  reduces to the form
\begin{equation}\label{hirota2}
  iu_t+\alpha(u_{xx}-2|u|^2u)+i\beta(u_{xxx}-6|u|^2u_x)=0,
\end{equation}
 where real parameters $\alpha$ and $\beta$ stand for the second dispersion and the third dispersion, respectively.  The equation
  (\ref{hirota2})  is integrable  system  which   is   the combination  of  complex mKdV  equation  and   the NLS equation.

   On account of the remarkable properties and the important role played in the scientific research,  much work on  a series of theoretical and practical work on various problems of this equation  has been  done.     Hirota obtained the exact $N$-envelope solitons by applying the bilinear direct method.
The relation between discrete surfaces  with constant negative gaussian curvature and the Hirota equation was considered in \cite{Alexander}. 
 The rogue wave solution and rational solution of the Hirota equation were further  studied  \cite{AnkiewiczA,YanZ}. 
 Nevertheless, as for long-time asymptotic analysis, the nonlinear steepest descent method developed by Deift and Zhou
    has been proved to be one of the most effective method \cite{PDEIFT}.  Based on the nonlinear steepest descent method,  many meaningful asymptotic analysis results have been investigated
    \cite{RN9,RN10,Grunert2009,MonvelCH,xu2015,xusp,fNLS,Liu3,SandRNLS,YF2021,CF2022}. For example, Huang  et al.  analyzed  the high order asymptotics  for the Hirota equation via the Deift-Zhou high order theory \cite{Huang}.
Guo et al.  first considered the long time asymptotic  behaviour  of the solution for the Hirota equation on the half line \cite{Guo}.
The  asymptotic  analysis  on  the high-order solitons was  discussed by Ling   \cite{Ling}.
Boutet de Monvel et al. discussed the Painlev\'{e}-type asymptotics for the Camassa-Holm equation by nonlinear steepest descent method \cite{Monvel}.    Charlier and   Lenells have carefully  considered the Airy and Painlev\'{e} asymptotics for the mKdV equation in \cite{Charlier}.  Recently,   Huang and  Zhang complete the extension from the Painlev\'{e} asymptotics analysis for the mKdV equation to that of the mKdV hierarchy\cite{Huanglin}.

To our knowledge, the Painlev\'{e} asymptotics  of  the Hirota equation in transition region  
for  the  Hirota equation  are still not presented yet. So in present paper, we focus on the long-time asymptotic behavior  for the Hirota equation in 
tansition region $\{ (x,t) \in \mathbb{R}^2 \left|
 \right. |\frac{x}{t} -  \frac{\alpha^2}{3\beta}|t^{2/3}  \leq M \}$ by applying the improved nonlinear steepest descent method.
 
 The organization of this paper is as follows: In Section 1, we first recall the construction of the corresponding Riemann-Hilbert problem and introduce the Painlev\'{e} region $\mathcal{P}$ related to the Hirota equation. In Section 2, we focus on the long-time asymptotic  analysis for the Hirota equation in the sector $\mathcal{P}_{\leq}$. First, we make an analytic approximation  for the scattering data. Next, we do a series of contour deformation to convert the Riemann-Hilbert problem into the solvable model problem. Based on the  above operations, we can finally obtain the asymptotic  results of the solutions for the defocusing Hirota equation. In section 4, we investigate the asymptotic behavior of the Hirota equation in the sector $\mathcal{P}_{\geq}$ using the same way as the last section. Finally, we   find that the final asymptotic   result can be given explicitly in terms of  the real-valued solutions of the Painlev\'{e}-$\mathrm{II}$ equation.

\section{Riemann-Hilbert problem}

We consider  the  Cauchy problem for defocusing Hirota equation 
\begin{equation}\label{HrtEq}
       \begin{aligned}
                & iu_t+\alpha(u_{xx}-2|u|^2u)+i\beta(u_{xxx}-6|u|^2u_x)=0,\ t>0, x\in \mathbb{R}, \\
                & u(x,0)=u_0(x) \in \mathcal{S} \mathbb{(R)},
       \end{aligned}
\end{equation}
where $\alpha\in \mathbb{R}$,  $\beta \in \mathbb{R}^{+}$. The Lax pair corresponding to \eqref{HrtEq} is given by 
\begin{equation}\label{Lax}
       \phi_x=P\phi,  \qquad \phi_t=Q\phi,
\end{equation}
where
\begin{equation}
       \begin{aligned}
            &\phi=\left(\begin{array}{lr}\phi_1\\ \phi_2\end{array}\right),\quad P=-ik\sigma_3+U,\quad U= \begin{pmatrix} 0&u\\ \overline{u}&0 \end{pmatrix},\\
            &Q=-4i\beta\sigma_3 k^3-2i\alpha\sigma_3 k^2 +V,\quad V=V_2 k^2+V_1 k+V_0,\\
            &V_2=4\beta U,\quad V_1= \begin{pmatrix}-2i \beta |u|^2&2i\beta u_x+2\alpha u\\
                -2i \beta \overline{u}_x+2\alpha \overline{u}&2i\beta|u|^2   \end{pmatrix},\\
             &V_0=  \begin{pmatrix} -i\alpha |u|^2+\beta(-u\overline{u}_x+u_x\overline{u})&i\alpha u_x -\beta(u_{xx}-2|u|^2u)\\-i \alpha \overline{u}_x+\beta(-\overline{u}_{xx}+2|u|^2\overline{u})&i\alpha |u|^2-\beta(-u\overline{u}_x+u_x\overline{u})\end{pmatrix},
       \end{aligned}
\end{equation}
and $\overline{u}$ denotes the  complex conjugate of  $u$. Here $\sigma_3$ is one of the Pauli matrices defined by
\begin{equation}
          \begin{aligned}
                \sigma_1=\begin{pmatrix} 0&1\\1&0 \end{pmatrix},  \quad \sigma_2= \begin{pmatrix} 0&i\\-i&0 \end{pmatrix},
                \quad \sigma_3=\begin{pmatrix}   1&0\\0&-1  \end{pmatrix}.
           \end{aligned}
\end{equation}
Considering the asymptotic property of  $\phi_{\pm}$, we make the transformation 
 $$\phi=\Phi e^{[-ikx-(4i\beta k^3+2i\alpha k^2)t]\sigma_3},$$
  we then  obtain  the equivalent Lax representation  
\begin{equation}\label{lax1}
           \begin{aligned}
             & \Phi_x+ik[\sigma_3,\Phi]=U\Phi,\\
             & \Phi_t+(2i\alpha k^2+4i\beta k^3)[\sigma_3,\Phi]=V\Phi.
           \end{aligned}
\end{equation}
We first define two solutions of the spectral problem  
\begin{eqnarray}\label{laxpair2}
                &&\Phi_-(x,t,k)=\mathrm{I}+\int_{-\infty}^x e^{-ik(x-y)\hat{\sigma}_3}U(y,t)\Phi_-(y,t,k)dy,\\
                &&\Phi_+(x,t,k)=\mathrm{I}-\int^{\infty}_x e^{-ik(x-y)\hat{\sigma}_3}U(y,t)\Phi_+(y,t,k)dy.
\end{eqnarray}
There exists a continuous matrix function $S(k)$ satisfying
\begin{equation}\label{jp1}
             \Phi_+(x,t,k)=\Phi_-(x,t,k)e^{-it\theta(k)\hat{\sigma}_3}S(k),
\end{equation}
where
\begin{equation}
      S(z) =  \begin{pmatrix}  s_{11}(k) &   s_{12}(k)  \\    s_{21}(k)     &   s_{22}(k)      \end{pmatrix},  \quad \theta(k)=4\beta k^3+2\alpha k^2+k\xi,\quad \xi=\frac{x}{t}.
\end{equation}
Moreover,    we have the critical  symmetry  $ S(k)=\sigma_1\overline{S(\overline{k})}\sigma_1$,  which implies that 
 matrix function   $S(k)$  can be written as the following form
\begin{equation}
         S(k) =  \begin{pmatrix}  a(k)   &  \overline{b(\overline{k})} \\   b(k)      &  \overline{ a (\overline{k})}      \end{pmatrix}.
\end{equation}
Based on the Abel formula, we obtain that   $ \det S(k) =1$.  Moreover, we have   $|a(k)|^2 = 1 + |b(k)|^2 $, $k \in \mathbb{R}$. For technical reasons, we assume here that  $a(k)$ has no singularity on the real axis.  Next, we define
\begin{equation}
              m(x,t,k)=\left\{\begin{array}{lr}\left(\dfrac{[\Phi_-]_1}{a(k)},[\Phi_+]_2\right),\quad \mathrm{Im}(k)>0,\\
             \left([\Phi_+]_1,\dfrac{[\Phi_-]_2}{\overline{a(\overline{k})}}\right),\quad \mathrm{Im}(k)<0.\end{array}\right.
\end{equation}
According to  the   relation \eqref{jp1}, we obtain  that $m(x,t,k)$ satisfies the following jump condition
\begin{equation}
           m_+(x,t,k)=m_-(x,t,k)J(x,t,k),  \quad k \in \mathbb{R},
\end{equation}
where
\begin{equation}\label{jpK}
              J(x,t,k)=\begin{pmatrix}  1-|r(k)|^2&-\overline{r(k)}e^{-2it\theta(k) }\\r(k)e^{2it\theta(k)}&1\end{pmatrix}, \quad r(k)=\frac{  \overline{b(\overline{k})} }{ a(k)}.
\end{equation}
The potential  $u(x,t)$   is given by
\begin{equation}\label{recstr}
           u(x,t)=2i\lim_{k\to\infty}(km)_{12},
\end{equation}
where $m$ is a $2 \times 2$ matrix-value function satisfying following Riemann-Hilbert problem:
\begin{theorem}
        Given $r(k)$, the function $m(x,t,k)$ satisfies the matrix Riemann-Hilbert problem as follows:\\
               $\bullet$ Analyticity: $m(x,t,k)$ is analytic in $k\in\mathbb{C}\setminus\mathbb{R}$.\\
               $\bullet$ Jump condition: $ m_+(x,t,k)=m_-(x,t,k)J(x,t,k),\ k\in \mathbb{R}.$  \\
               $\bullet$ Asymptotic property:  $ m(x,t,k)\to\mathrm{I},\ as\ k \to \infty.$
\end{theorem}
\noindent  The signature table of $\mathrm{Re}(i\theta)$ is shown in Figure 1, and stationary points as are given by
\begin{equation}\label{stationary}
k_1=\frac{-\alpha-\sqrt{\alpha^2-3\beta \xi}}{6\beta},\quad k_2=\frac{-\alpha+\sqrt{\alpha^2-3\beta \xi}}{6\beta}.
\end{equation}
\begin{figure}[H]
      \centering
          \includegraphics[width=6.5cm]{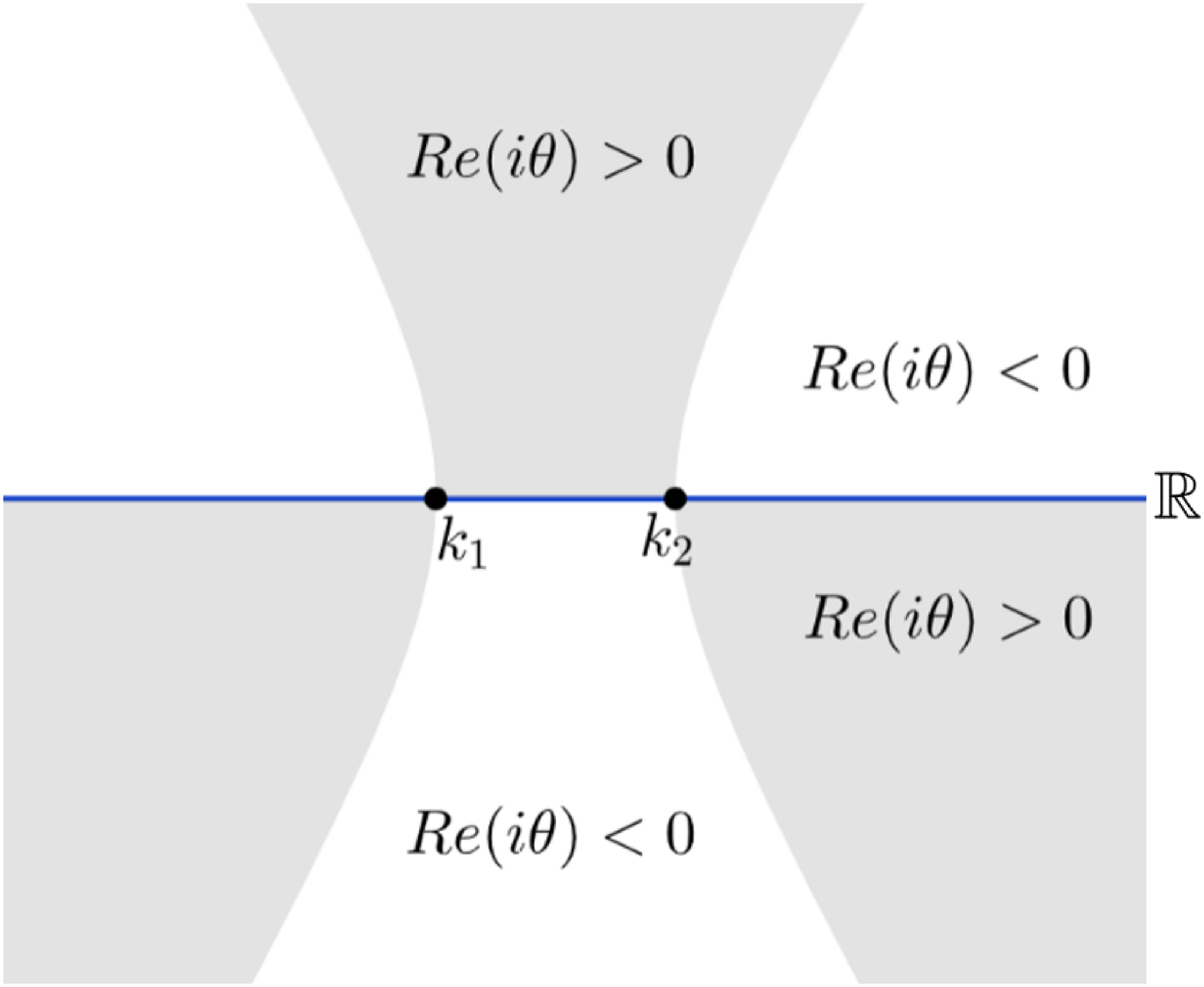}\qquad
          \includegraphics[width=6.5cm]{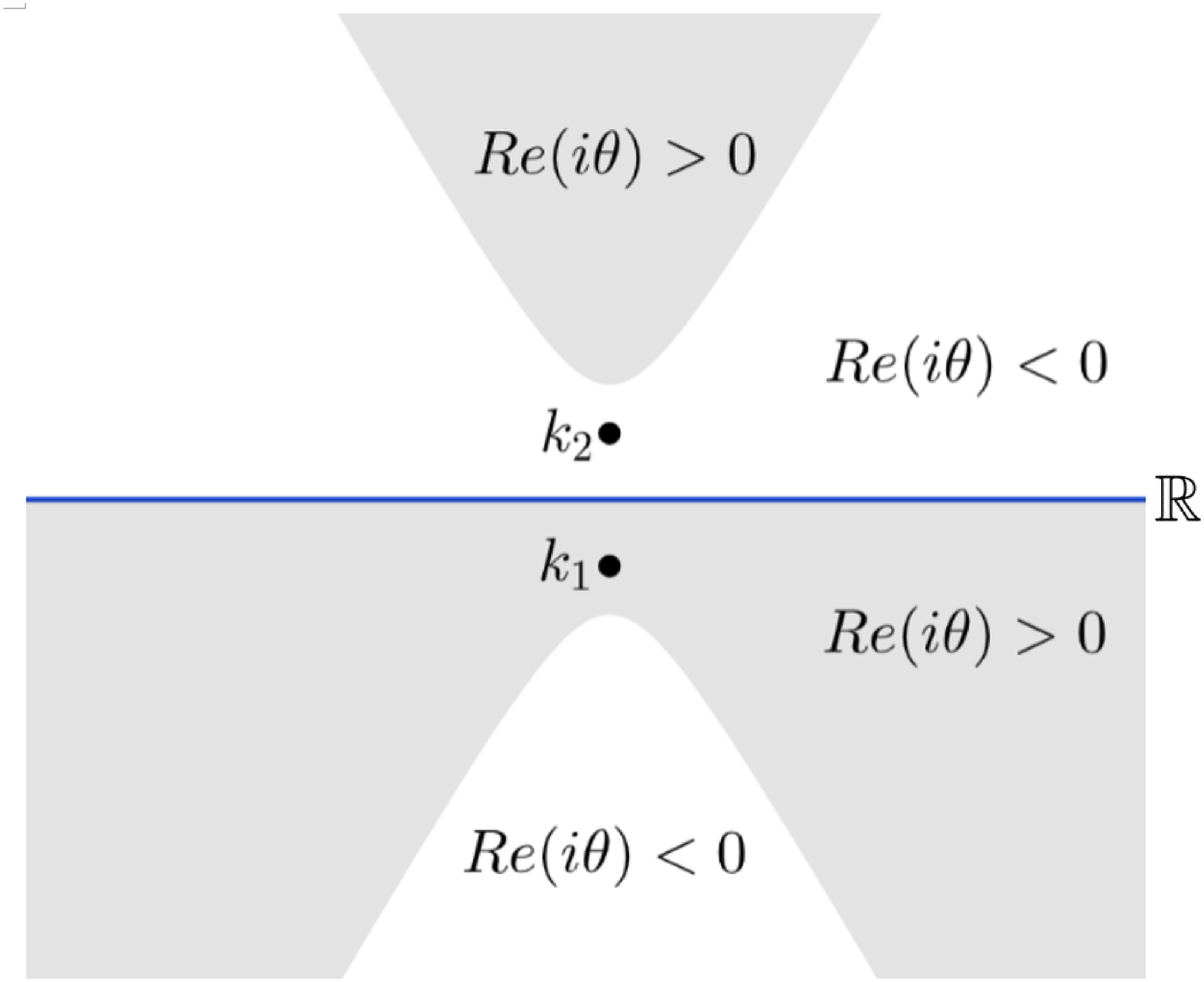}\\
          \caption{ \footnotesize (colour online). The critical points $k_1, k_2$ in the complex $k$-plane in the case of $\xi<\frac{\alpha^2}{3\beta}$ (left) and $\xi>\frac{\alpha^2}{3\beta}$ (right). The regions where $\mathrm{Re}(i\theta)> 0$ and $\mathrm{Re}(i\theta)< 0$ are shaded and white, respectively.}\label{fg1}
\end{figure}
\noindent Based on the signature table, the jump matrix $J(x,t,k)$ has following triangular factorization
\begin{equation}
        J(x,t,k)=\left(\begin{array}{cc}1-|r(k)|^2&-\overline{r(k)}e^{-2it\theta(k) }\\r(k)e^{2it\theta(k)}&1\end{array}\right)=
           \left(\begin{array}{cc}1&-\overline{r(k)}e^{-2it\theta(k) }\\0&1\end{array}\right)\left(\begin{array}{cc}1&0 \\r(k)e^{2it\theta(k)}&1\end{array}\right).\nonumber
\end{equation}
We aim to find the asymptotics of $u(x,t)$ in the transition region defined as
\begin{equation}
           \mathcal{P}:=\left\{  (x,t)\in \mathbb{R}^2, ~  0< |\xi-\frac{\alpha^2}{3\beta} |t^{2/3}< M      \right\},
\end{equation}
where $M  > 0$ is a constant.  We use  the  following notations
\begin{equation}
              \mathcal{P}_{\leq}:=\mathcal{P}\cap\left\{\xi\leq\frac{\alpha^2}{3\beta}\right\},\quad
              \mathcal{P}_{\geq}:=\mathcal{P}\cap\left\{\xi\geq\frac{\alpha^2}{3\beta}\right\}
\end{equation}
to denote the left and right halves of $\mathcal{P}$, respectively.

\section{Asymptotics in Sector $\mathcal{P}_{\leq}$}

Suppose $(x,t)\in \mathcal{P}_{\leq}$. In this region, the two stationary points $k_1$, $k_2$ defined by \eqref{stationary} 
are real and close to $-\frac{\alpha}{6\beta}$ at least as the speed of $t^{-\frac{1}{3}}$ as $t \to \infty$.
\subsection{Analytical  approximation}

Let $\Gamma \in \mathbb{C}$ denote the contour $\Gamma=(\cup_{j=1}^4 l_j)\cup\mathbb{R}$ oriented to the right as in Figure 2, where
\begin{equation}
\begin{aligned}
l_1&:=\{k_1+le^{\frac{5\pi i}{6}},\ l>0\},   &  l_2:=\{k_2+le^{\frac{\pi i}{6}},\ l>0\},\\
l_3&:=\{k_1+le^{-\frac{5\pi i}{6}},\ l>0\},   & l_4:=\{k_2+le^{-\frac{\pi i}{6}},\ l>0\},
\end{aligned}
\end{equation}
and let 
\begin{align}
&D =\{\arg (k-k_1) \in(5 \pi / 6, \pi)\} \cup\{\arg (k-k_2) \in(0, \pi / 6)\},\nonumber\\
&D^*=\{\arg (k-k_1) \in(-\pi,-5 \pi / 6)\} \cup\{\arg (k-k_2) \in(-\pi / 6,0)\} \nonumber
\end{align}
 denote the open subsets shown in the same figure.
\begin{figure}[H]
  \centering
  \includegraphics[width=7cm]{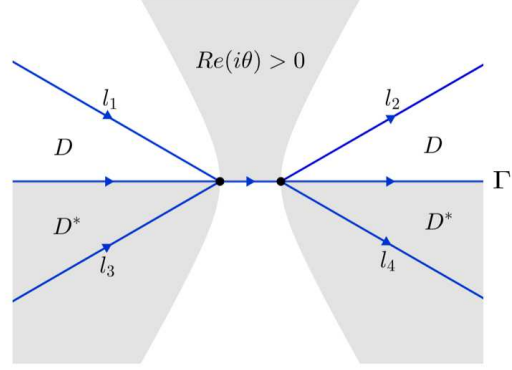}\\
            \caption{ \footnotesize (colour online). The contour $\Gamma$ and the sets $D$ and $D^*$ in the case of $\mathcal{P}_{\leq}$. The region where $\mathrm{Re}(i\theta)>0$ is shaded.}
\end{figure}
\noindent
 Then we have an analytical  approximation   for $r(k)$ as follows:
\begin{proposition}
There exists a decomposition
\begin{equation}
r(k)=r_a(x,t,k)+r_r(x,t,k),\quad k \in (-\infty,k_1)\cup(k_2,\infty),
\end{equation}
where $r_a$ and $r_r$ satisfy the following properties:\\
(i)\ For $(x,t)\in \mathcal{P}_{\leq}$, $r_a(x,t,k)$ is defined and continuous for $k\in \overline{D}$ and analytic for $k\in D$.\\
(ii)\ The function $r_a(x,t,k)$ satisfies
\begin{equation}
|r_a(x,t,k)|\leq\frac{C}{1+|k|^2}e^{\frac{t}{4}|\text{Re}(2i\theta(k))|},\ k \in \overline{D},
\end{equation}
and
\begin{equation}
|r_a(x,t,k)-r(k_j)|\leq C|k-k_j|e^{\frac{t}{4}|\text{Re}(2i\theta(k))|},\ k\in \overline{D},\ j=1,2.
\end{equation}
(iii)\ The $L^1$, $L^2$ and $L^{\infty}$ norms of the function $r_r(x,t,\cdot)$ on $\mathbb{R}\setminus (k_1,k_2) $ are $O(t^{-\frac{3}{2}})$ as $t\to\infty$ uniformly for $(x,t)\in \mathcal{P}_{\leq}$.
\end{proposition}
\begin{proof}See \cite{Lns2017}, Lemma 4.8 for more details.\end{proof}
\subsection{Contour deformation}
 In what follows, we  perform the contour deformation as follows:
\begin{equation}\label{chg1}
m^{(1)}(x,t,k)=m(x,t,k)\times\left\{\begin{aligned}
&\begin{pmatrix}1&0\\-r_a(k)e^{2it\theta}&1\end{pmatrix}, &&k \in D,\\
&\begin{pmatrix}1&-\overline{r_a(\overline{k})}e^{-2it\theta}\\0&1\end{pmatrix}, &&k \in D^*,\\
&\ \mathrm{I}, &&\text{elsewhere}.\end{aligned}\right.
\end{equation}
Thus, we find $m^{(1)}(x,t,k)$ satisfies the new RH problem
\begin{equation}\label{jmp1}
m^{(1)}_+(x,t,k)=m^{(1)}_-(x,t,k)J^{(1)}(x,t,k),
\end{equation}
where
\begin{equation}
J^{(1)}(x,t,k)=\left\{\begin{aligned}&\left(\begin{array}{cc}
1&0\\r_a(k)e^{2it\theta}&1\end{array}\right), &&k\in l_1\cup l_2,\\
&\left(\begin{array}{cc}
1&-\overline{r_a(\overline{k})}e^{-2it\theta}\\0&1\end{array}\right), &&k\in l_3\cup l_4,\\
&\left(\begin{array}{cc}1-|r(k)|^2&-\overline{r(k)}e^{-2it\theta(k) }\\r(k)e^{2it\theta(k)}&1\end{array}\right), &&k \in (k_1,k_2),\\
&\left(\begin{array}{cc}1-|r_r(k)|^2&-\overline{r_r(k)}e^{-2it\theta(k) }\\r_r(k)e^{2it\theta(k)}&1\end{array}\right), &&k \in \mathrm{R}/(k_1,k_2).
\end{aligned}\right.
\end{equation}
\subsection{Local  model}
Considering $ \mid \xi-\frac{\alpha^2}{3\beta}\mid t^{\frac{2}{3}}\leq C,$
it is obvious that for $t\to +\infty$, $k_j\to -\frac{\alpha}{6\beta}$. The phase function can be approximated as
$t\theta(k)=t\theta(-\frac{\alpha}{6\beta})+s\hat{k}+ \frac{4}{3}\hat{k}^3,$
where
\begin{equation}\label{sk}
s=(3\beta)^{-1/3}(\xi-\frac{\alpha^2}{3\beta})t^{2/3}, \quad \hat{k}=(3\beta t)^{1/3}(k+\frac{\alpha}{6\beta}).
\end{equation}
The coefficients in above formula have been fittingly chosen such that the form of the scaled phase function is the same as that of the RH problem for the Painlev\'{e} $\amalg$ equation. For a fixed $\varepsilon>0$, let $D_{\varepsilon}(-\frac{\alpha}{6 \beta})=\{k \in \mathbb{C}|| k+\frac{\alpha}{6 \beta} \mid<\varepsilon\}$, and let $\Gamma^{\varepsilon}=(\Gamma \cap D_{\varepsilon}(-\frac{\alpha}{ 6 \beta})) \backslash((-\infty, k_{1}) \cup (k_{2},\infty))$. Next, we define
\begin{equation}
m^{(2)}(s, t, \hat{k})=m^{(1)}(x, t, k) e^{-i t\theta\left(-\frac{\alpha}{6 \beta}\right) \sigma_{3}},\quad k\in D_\varepsilon(-\frac{\alpha}{6\beta})\setminus\Gamma.
\end{equation}
Then the new jump matrix $J^{(2)}$ can be approximated as follows:
\begin{equation}
\begin{aligned}
        J^{(2)}(s, t, \hat{k})
            \to  \left\{\begin{array}{lr}
             \left(\begin{array}{cc}
                       1&0\\r\left(-\frac{\alpha}{6 \beta}\right) e^{2 i\left(s \hat{k}+\frac{4}{3} \hat{k}\right)} & 1\end{array} \right),\quad\quad\ \ k \in (\Gamma^{\varepsilon})_1,\\
             \left(\begin{array}{cc}
                   1 & -\overline{r\left(-\frac{\alpha}{6 \beta}\right)}e^{-2 i\left(s \hat{k}+\frac{4}{3} \hat{k}^{3}\right)}\\
                    0 & 1\end{array}\right),\quad k \in (\Gamma^{\varepsilon})_{2}, \\
            \left(\begin{array}{cc}
                    1 & -\overline{r\left(-\frac{\alpha}{6 \beta}\right)}e^{-2 i\left(s \hat{k}+\frac{4}{3} \hat{k}^{3}\right)}\\
                    0 & 1\end{array}\right)\left(\begin{array}{cc}
                    1&0\\r\left(-\frac{\alpha}{6 \beta}\right) e^{2 i\left(s \hat{k}+\frac{4}{3} \hat{k}\right)} & 1\end{array} \right),\quad k \in (\Gamma^{\varepsilon})_3, \end{array}\right.
\end{aligned}
\end{equation}
which is consistent with the jump matrix $\hat{J}$ defined by the model RH problem in terms of the solution of the Painlev$\'{e}$ $\mathrm{II}$ equation with $ s = i \left| r  (-\frac{\alpha}{6\beta} ) \right|$, $\left| r  (-\frac{\alpha}{6\beta} ) \right| <1$ in Appendix A.2. of \cite{Charlier}. We write $\Gamma^{\varepsilon}=\mathop{\bigcup}\limits_{j=1}^{3}\Gamma_j^{\varepsilon}$, where $\Gamma_j^{\varepsilon}$ denotes the part of $\Gamma^{\varepsilon}$ that maps into $j$, see Figure 3.
\begin{figure}[H]
  \centering
  \includegraphics[width=6cm]{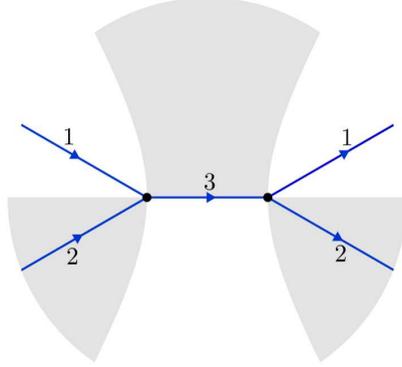}\\
 \caption{ \footnotesize (colour online). The contour $\Gamma^{\varepsilon}=\mathop{\bigcup}\limits_{j=1}^{3}\Gamma_j^{\varepsilon}$.}
\end{figure}
\noindent Thus we expect that $m^{(1)}(x,t,k)$ in $D_\varepsilon(-\frac{\alpha}{6\beta})$ approaches the solution $m^r(x,t,k)$ defined by $
m^{r}(x,t,k)=\mathrm{e}^{-\mathrm{i} t \theta\left(-\frac{\alpha}{6\beta}\right) \hat{\sigma}_{3}} \hat{m}(\rho, s, \hat{k})
$
as $t \to \infty$, where $\hat{m}(\rho, s, \hat{k})$ is  a  model RH problem in terms of the solution of the Painlev\'{e} $\mathrm{II}$ equation in Appendix A.2. of \cite{Charlier}.

\begin{proposition}For each $(x,t) \in \mathcal{P}_{\leq}$, $m^{r}(x, t, k)$ is an analytic function of $k \in D_{\varepsilon}(-\frac{\alpha}{6\beta}) \backslash \Gamma^{\varepsilon}$ such that $\left|m^{r}(x, t, k)\right| \leq C$. Across $\Gamma^{\varepsilon}$,  $m^{r}(x, t, k)$ has the jump condition $m_{+}^{r}=m_{-}^{r} J^{r}$, where the jump matrix $J^r$ satisfies
\begin{equation}
\|J^{(1)}-J^{r}\|_{L^{1} \cap L^{2} \cap L^{\infty}\left(\Gamma^{\varepsilon}\right)} \leq C t^{-\frac{1}{3}}.
\end{equation}
Furthermore, as $t \to \infty$,
$$
\left\|\left(m^{r}\right)^{-1}(x, t, k)-I\right\|_{L^{\infty}\left(\partial D_{\varepsilon}(-\frac{\alpha}{6 \beta})\right)}=O\left(t^{-\frac{1}{3}}\right),
$$
and
\begin{equation}\label{314}
\frac{1}{2 \pi i} \int_{\partial D_{\varepsilon}(-\frac{\alpha}{6 \beta})}\left(\left(m^{r}\right)^{-1}(x, t, k)-I\right) d k=-\frac{m_{1}^{r}(s)}{(3 \beta t)^{1 / 3}}+O\left(t^{-\frac{2}{3}}\right),
\end{equation}
where
\begin{equation}\label{m1}
m_{1}^{r}(s)=\left(\begin{array}{cc}
\frac{i}{2} \int^{s} y^{2}(\zeta) d \zeta &  \frac{i}{2} e^{-2 i t \theta\left(-\frac{\alpha}{6 \beta}\right)-i \gamma}  y(s) \\
-\frac{i}{2} e^{2 i t \theta\left(-\frac{\alpha}{6 \beta}\right)+ i \gamma}   y(s) & -\frac{i}{2} \int^{s} y^{2}(\zeta) d \zeta
\end{array}\right), \quad \gamma=\arg r\left(-\frac{\alpha}{6 \beta}\right),
\end{equation}
where $y(s)$ is a real-valued solution to Painlev\'{e} $\mathrm{II}$ equation
\begin{equation}
       y_{ss}(s) = s  y(s)  +  2  y(s)^3,
\end{equation}
which is fixed by its asymptotics as $s \to +\infty$,
\begin{equation}
y(s) \sim - \frac{1}{2 \sqrt{\pi}} \left|r\left(-\frac{\alpha}{6 \beta}\right)\right|  s^{-1/4} \exp(-\frac{2}{3} s^{3/2}).
\end{equation}
\end{proposition}
\begin{proof}
The proof is similiar to the proof of  Lemma 4.2 in \cite{Charlier}.
\end{proof}
\noindent  Assume that the boundary of $D_{\varepsilon}(-\frac{\alpha}{6\beta})$ is oriented counterclockwise. Define the approximate solution $m^{app}$ by
\begin{equation}\label{316}
m^{app}(x,t,k)= \begin{cases}m^{r}(x,t,k), & k\in D_{\varepsilon}(-\frac{\alpha}{6\beta}), \\  I, & \text { elsewhere. }\end{cases}
\end{equation}
Then the error function $E_r(x,t,k)$ defined by
\begin{equation}\label{Er}
E_r(x,t,k)=m^{(1)}\left(m^{a p p}\right)^{-1}
\end{equation}
satisfies a small-norm RH problem with the jump relation ${\left(E_r\right)}_{+}={\left(E_r\right)}_{-}J_r$ across $\breve{\Gamma}=\Gamma\cup\partial D_{\varepsilon}(-\frac{\alpha}{6\beta})$, where the jump matrix $J_r$ is given by
\begin{equation}\label{jmpr}
J_r= \begin{cases}m_{-}^{(a p p)} J^{(1)}\left(m_{+}^{(a p p)}\right)^{-1}, \quad &k \in \Gamma\cap D_{\varepsilon}(-\frac{\alpha}{6\beta}),\\
m_{-}^{(a p p)}\left(m_{+}^{(a p p)}\right)^{-1}, \quad &k \in\partial
D_{\varepsilon}(-\frac{\alpha}{6\beta}), \\
J^{(1)}, \quad &k \in \Gamma \backslash \overline{D_{\varepsilon}(-\frac{\alpha}{6\beta})}.
\end{cases}
\end{equation}
We now denote $\breve{\Gamma}$ as $ \breve{\Gamma}=\breve{\Gamma}_{1} \cup \breve{\Gamma}_{2} \cup \breve{\Gamma}_{3} \cup \breve{\Gamma}_{4}$, where $\breve{\Gamma}_{1}=\breve{\Gamma} \backslash (\mathbb{R} \cup \overline{D_{\varepsilon} (-\frac{\alpha}{6\beta})})$,  $\breve{\Gamma}_{2}=\mathbb{R} \backslash\left[k_{1}, k_{2}\right]$, $\breve{\Gamma}_{3}=\partial D_{\varepsilon}(-\frac{\alpha}{6\beta})$,   $\breve{\Gamma}_{4}=\Gamma^{\varepsilon}.$
\begin{figure}[H]
  \centering
  \includegraphics[width=6cm]{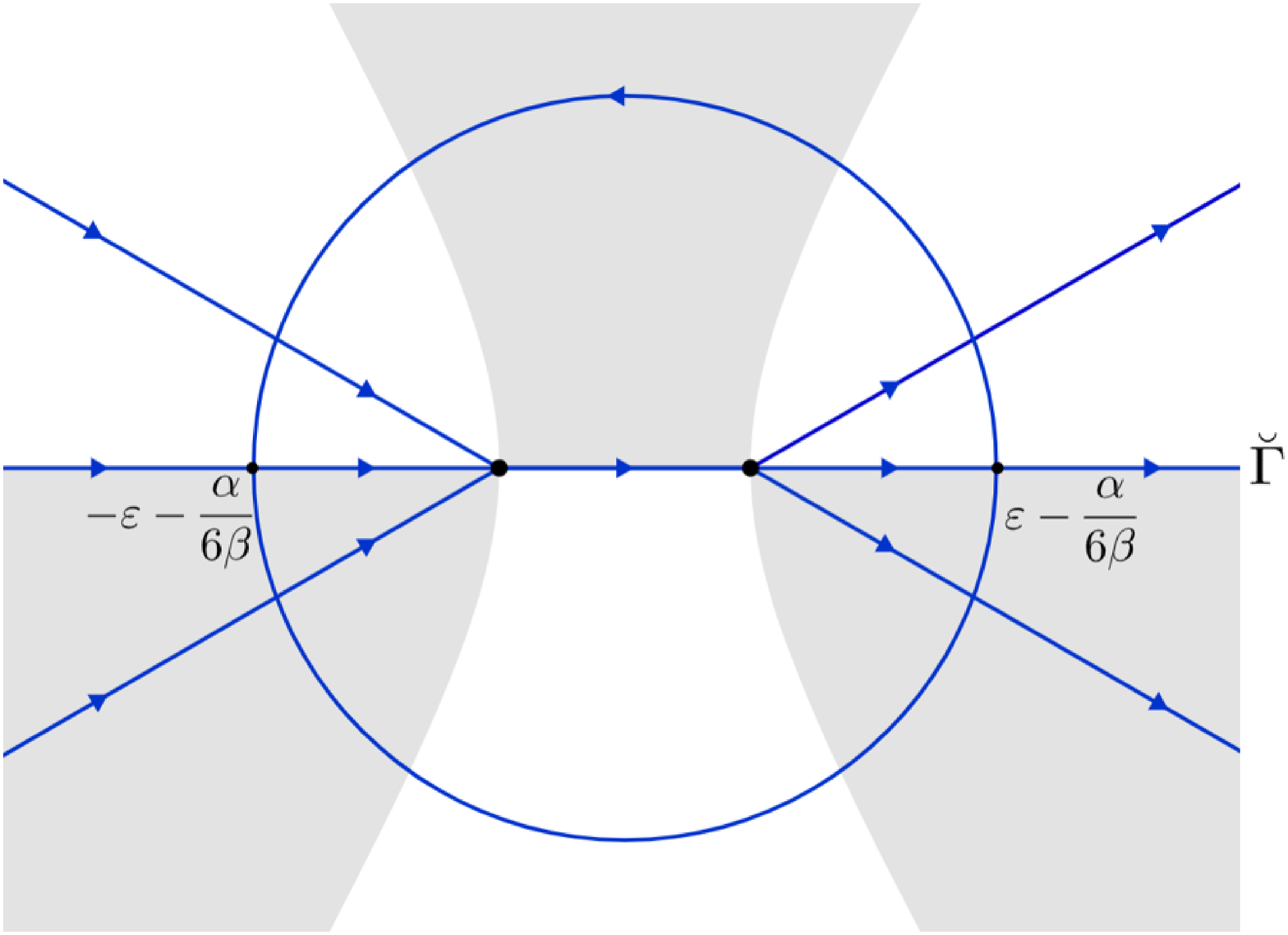}\\
 \caption{ \footnotesize (colour online). The contour $\breve{\Gamma}$ in the case of $\mathcal{P}_{\leq}$.}
\end{figure}

\begin{proposition}Let $w_r=J^r-I .$ For $(x, t) \in \mathcal{P}_{\leq}$, the following estimates hold:
\begin{equation}
\|w_r\|_{L^{1} \cap L^{2} \cap L^{\infty}\left(\breve{\Gamma}_{1}\right)} \leq C e^{-c t},
\end{equation}
\begin{equation}
\|w_r\|_{L^{1} \cap L^{2} \cap L^{\infty}\left(\breve{\Gamma}_{2}\right)} \leq C t^{-\frac{3}{2}},
\end{equation}
\begin{equation}\label{321}
\|w_r\|_{L^{1} \cap L^{2} \cap L^{\infty}\left(\check{\Gamma}_{3}\right)} \leq C t^{-\frac{1}{3}},
\end{equation}
\begin{equation}
\|w_r\|_{L^{1} \cap L^{2} \cap L^{\infty}\left(\check{\Gamma}_{4}\right)} \leq C t^{-\frac{1}{3}}.
\end{equation}
\end{proposition}

\subsection{Asymptotics of  the solution}
In  this subsection,  we will derive the asymptotics formula of the solution for the Hirota equation in Sector $\mathcal{P}_{\leq}$. Firstly,   we set $\breve{C}$ as the Cauchy operator associated with $\breve{\Gamma}$ and let $\breve{C}_{w_r} f := \breve{C}_{-}(f w_r)$.  Therefore,  matrix error function $E_{r}(k)$ can be rewritten as
$$
E_r(x, t, k)=I+\frac{1}{2 \pi \mathrm{i}} {\int}_{\breve{\Gamma}} \frac{(\mu_r w_r)(x, t, \zeta)}{\zeta-k} \mathrm{~d} \zeta,
$$
where the $2 \times 2$ matrix-valued function $\mu_r(x, t, k)$ is defined by $\mu_r=\mathrm{I}+\breve{C}_{w_r} (\mu_r)$. Moreover, using the Neumann series, $\mu_r(x, t, k)$ satisfies
$$
\|\mu_r-I\|_{L^2(\check{\Gamma})}=O(t^{-\frac{1}{3}}),\quad t\to \infty.
$$
It   follows that
\begin{equation}\label{323}
\lim _{k \to +\infty}k(m(x, t, k)-I)=-\frac{1}{2 \pi \mathrm{i}}\int_{\breve{\Gamma}}(\mu_r w_r)(x, t, \zeta) \mathrm{d}\zeta.
\end{equation}
By\eqref{314}, \eqref{316}, \eqref{jmpr} and \eqref{321}, the contribution from $\partial D_{\varepsilon}(-\frac{\alpha}{6\beta})$ to the right-hand side of \eqref{323} is
\begin{equation}
\begin{aligned}
-\frac{1}{2 \pi \mathrm{i}}\int_{\partial D_{\varepsilon}}(\mu_r w_r)(x, t, \zeta) \mathrm{d}\zeta&=-\frac{1}{2 \pi \mathrm{i}} \int_{\partial D_{\varepsilon}(-\frac{\alpha}{6\beta})} w_r \mathrm{~d} \zeta-\frac{1}{2 \pi \mathrm{i}} \int_{\partial D_{\varepsilon}(-\frac{\alpha}{6\beta})}(\mu_r-I) w_r \mathrm{~d} \zeta \\
&=\frac{m_{1}^{r}(s)}{(3\beta t)^{1/3}}+O\left(t^{-\frac{2}{3}}\right) .
\end{aligned}
\end{equation}
The contributions from $\breve{\Gamma}_{1}, \breve{\Gamma}_{2}$ and $\breve{\Gamma}_{4}$ to the right-hand side of \eqref{323} are $O\left(e^{-c t}\right), O\left(t^{-3 / 2}\right)$ and $O\left(t^{-1 / 3}\right)$, respectively. Recalling the reconstructional formula \eqref{recstr}, \eqref{323} and the definition of $\phi$, we immediately obtain the asymptotic formula of $u(x, t)$ as follows
\begin{equation}
  u(x,t)=  -\frac{1}{(3\beta t)^{1/3}} \exp(-2i \theta(-\frac{\alpha}{ 6 \beta}) - i \gamma )  y(s) +  O\left(t^{-\frac{2}{3}}\right),
\end{equation}
where $s=(3\beta)^{-1/3}(\xi-\frac{\alpha^2}{3\beta})t^{2/3}$. So far, we  completed  the  long time asymptotic analysis of the Hirota  equation in the space-time region  $ \left\{  (x,t)\in \mathbb{R}^{2}  \bigg| -M  \leq ( \frac{x}{t} -  \frac{\alpha^2}{3\beta^2} ) t^{2/3}  \leq 0   \right\}.$

\section{Asymptotics in Sector $\mathcal{P}_{\geq}$}
We now consider the asymptotics in sector $\mathcal{P}_{\geq}$. In this sector, the two stationary points $k_1, k_2$ are complex number and approach to $-\frac{\alpha}{6\beta}$ as the speed of $t^{-\frac{1}{3}}$ as $t \to \infty$.
As in Section 3, we first decompose $r$ into two parts. In this part, we  define the contour $\Gamma$ and the open subsets $D, D^*$ as in Figure 5.
\begin{figure}[H]
  \centering
  \includegraphics[width=7.5cm]{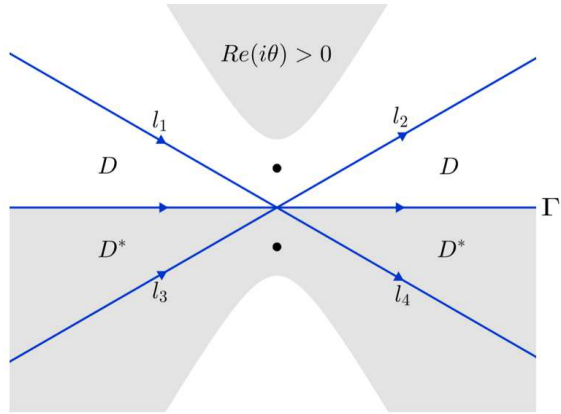}\\
 \caption{ \footnotesize (colour online). The contour $\Gamma$ and the sets $D$ and $D^*$ in the case of $\mathcal{P}_{\geq}$.\\
The region where $\mathrm{Re}(i\theta)>0$ is shaded.}
\end{figure}

\begin{proposition}
There exists a decomposition
\begin{equation}
r(k)=r_a(x,t,k)+r_r(x,t,k),\quad k \in \mathbb{R},
\end{equation}
where $r_a$ and $r_r$ satisfy the following properties:\\
(i)\ For $(x,t)\in \mathcal{P}_{\geq}$, $r_a(x,t,k)$ is defined and continuous for $k\in \overline{D}$ and analytic for $k\in D$.\\
(ii)\ The function $r_a(x,t,k)$ satisfies
\begin{equation}
|r_a(x,t,k)|\leq\frac{C}{1+|k|^2}e^{\frac{t}{4}|\text{Re}(2i\theta(k))|},\quad k \in \overline{D},
\end{equation}
and
\begin{equation}
|r_a(x,t,k)-r(-\frac{\alpha}{6 \beta})|\leq C|k+\frac{\alpha}{6 \beta}|e^{\frac{t}{4}|\text{Re}(2i\theta(k))|},\quad k\in \overline{D}.
\end{equation}
(iii)\ The $L^1$, $L^2$ and $L^{\infty}$ norms of the function $r_r(x,t,\cdot)$ on $\mathbb{R}$ are $O(t^{-\frac{3}{2}})$ as $t\to\infty$ uniformly for $(x,t)\in \mathcal{P}_{\geq}$.
\end{proposition}
\noindent Using this decomposition of $r(x,t,k)$, we define $m^{(1)}$ as \eqref{chg1}, then the jump matrix $J^{(1)}$ in \eqref{jmp1} changes into
\begin{equation}
J^{(1)}(x,t,k)=\left\{\begin{aligned}&\left(\begin{array}{cc}
1&0\\r_a(k)e^{2it\theta}&1\end{array}\right), &&k\in l_1\cup l_2,\\
&\left(\begin{array}{cc}
1&-\overline{r_a(\overline{k})}e^{-2it\theta}\\0&1\end{array}\right), &&k\in l_3\cup l_4,\\
&\left(\begin{array}{cc}1-|r_r(k)|^2&-\overline{r_r(k)}e^{-2it\theta(k) }\\r_r(k)e^{2it\theta(k)}&1\end{array}\right), &&k \in \mathrm{R}.
\end{aligned}\right.
\end{equation}
For $\mid \xi-\frac{\alpha^2}{3\beta}\mid t^{\frac{2}{3}}\geq C,$ as in Section 3, the phase function can also be approximated as $t\theta(k)=t\theta(-\frac{\alpha}{6\beta})+s\hat{k}+
\frac{4}{3}\hat{k}^3,$ and $s, \hat{k}$ are given by \eqref{sk}.
Let $\Sigma^{\varepsilon}=(\Gamma \cap D_{\varepsilon}(-\frac{\alpha}{b\beta})) \backslash\mathbb{R}$. Define
\begin{equation}
m^{(2)}(s, t, \hat{k})=m^{(1)}(x, t, k) e^{-i t\theta\left(-\frac{\alpha}{6 \beta}\right) \sigma_{3}},\quad k\in D_\varepsilon(-\frac{\alpha}{6\beta})\setminus\Gamma.
\end{equation}
We write $\Sigma^{\varepsilon}=\Sigma_1^{\varepsilon}\cup\Sigma_2^{\varepsilon}$, where $\Sigma_j^{\varepsilon}$ denotes the part of $\Sigma^{\varepsilon}$ that maps into $j$, see Figure 6.
Then the jump matrix $J^{(2)}$ can be approximated as
\begin{equation}
J^{(2)}(s, t, \hat{k})=\left\{\begin{aligned}
&\left(\begin{array}{cc}
1&0\\r\left(-\frac{\alpha}{6 \beta}\right) e^{2 i\left(s \hat{k}+\frac{4}{3} \hat{k}\right)} & 1\end{array} \right), &&k \in \Sigma^{\varepsilon}_1,\\
&\left(\begin{array}{cc}
1 & -\overline{r\left(-\frac{\alpha}{6 \beta}\right)}e^{-2 i\left(s \hat{k}+\frac{4}{3} \hat{k}^{3}\right)}\\
0 & 1\end{array}\right), &&k \in \Sigma^{\varepsilon}_{2}.  \end{aligned}\right.
\end{equation}
Thus we expect that $m^{(1)}(x,t,k)$ in $D_\varepsilon(-\frac{\alpha}{6\beta})$ approaches the solution $m^r(x,t,k)$ defined by $m^{r}(x,t,k)=\mathrm{e}^{-\mathrm{i} t \theta\left(-\frac{\alpha}{6\beta}\right) \hat{\sigma}_{3}} \hat{m}(\rho, s, \hat{k})$ as $t \to \infty$, where $\hat{m}(\rho, s, \hat{k})$ is the solution of the model RH problem.
\begin{figure}[H]
  \centering
  \includegraphics[width=6cm]{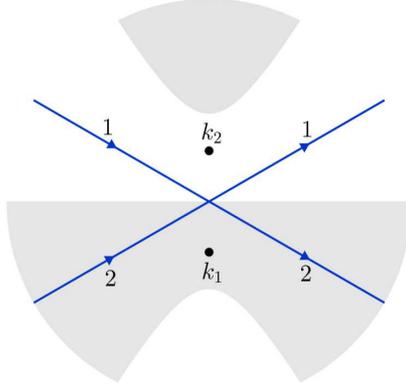}\\
 \caption{ \footnotesize (colour online). The contour $\Sigma^{\varepsilon}=\Sigma_1^{\varepsilon}\cup\Sigma_2^{\varepsilon}$.}
\end{figure}
\begin{proposition}For each $(x,t) \in \mathcal{P}_{\geq}$, $m^{r}(x, t, k)$ is an analytic function of $k \in D_{\varepsilon}(-\frac{\alpha}{6\beta}) \backslash \Gamma^{\varepsilon}$ such that $\left|m^{r}(x, t, k)\right| \leq C$. Across $\Gamma^{\varepsilon}, m^{r}(x, t, k)$ has the jump condition $m_{+}^{r}=m_{-}^{r} J^{r}$, where the jump matrix $J^r$ satisfies
\begin{equation}
\|J^{(1)}-J^{r}\|_{L^{1} \cap L^{2} \cap L^{\infty}\left(\Gamma^{\varepsilon}\right)} \leq C t^{-\frac{1}{3}}.
\end{equation}
Furthermore, as $t \to \infty$,
$$
\left\|\left(m^{r}\right)^{-1}(x, t, k)-\mathrm{I}\right\|_{L^{\infty}\left(\partial D_{\varepsilon}(-\frac{\alpha}{6 \beta})\right)}=O\left(t^{-\frac{1}{3}}\right),
$$
and
\begin{equation}
\frac{1}{2 \pi i} \int_{\partial D_{\varepsilon}(-\frac{\alpha}{6 \beta})}\left(\left(m^{r}\right)^{-1}(x, t, k)-\mathrm{I}\right) d k=-\frac{m_{1}^{r}(s)}{(3 \beta t)^{1 / 3}}+O\left(t^{-\frac{2}{3}}\right),
\end{equation}
where $m_{1}^{r}(s)$ is defined by \eqref{m1}.
\end{proposition}
\noindent Define $E_r(x,t,k)$ by \eqref{Er}, then $E_r$ satisfies the RH problem with the jump matrix $J_r$ given by \eqref{jmpr}. Denote $\breve{\Gamma}$ as $
\breve{\Gamma}=\Gamma\cup\partial
D_{\varepsilon}(-\frac{\alpha}{6\beta})=\breve{\Gamma}_{1} \cup \mathbb{R} \cup \partial D_{\varepsilon}(-\frac{\alpha}{6\beta}) \cup \Sigma^{\varepsilon}
$, $ \breve{\Gamma}_{1}=\breve{\Gamma} \backslash (\mathbb{R} \cup D_{\varepsilon} (-\frac{\alpha}{6\beta})).$
\begin{figure}[H]
  \centering
  \includegraphics[width=6cm]{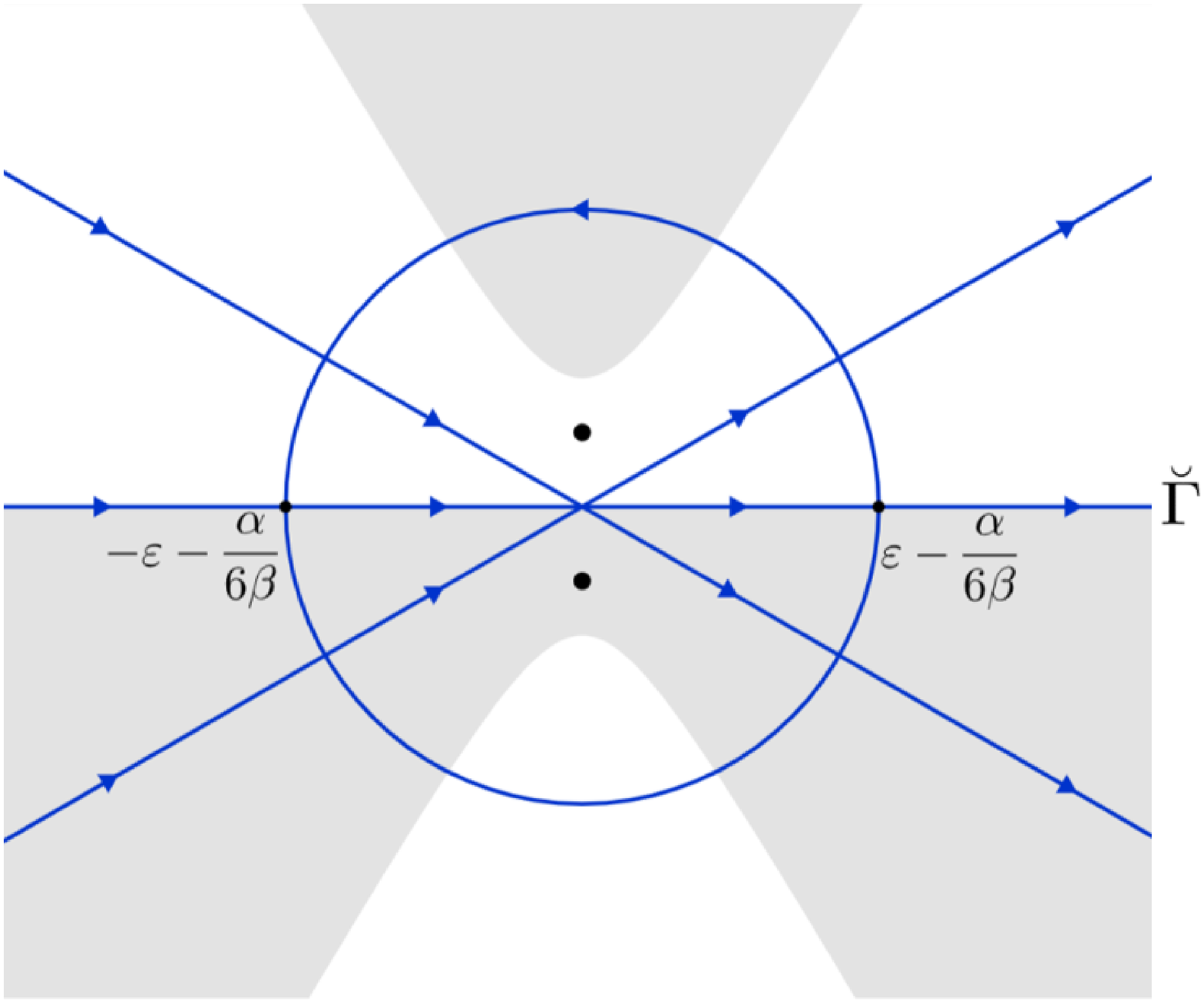}\\
 \caption{ \footnotesize (colour online). The contour $\breve{\Gamma}$ in the case of $\mathcal{P}_{\geq}$.}
\end{figure}

\begin{proposition}Let $w_r=J^r-I .$ For $(x, t) \in \mathcal{P}_{\geq}$, the following estimates hold:
\begin{equation}
\begin{aligned}
&\|w_r\|_{L^{1} \cap L^{2} \cap L^{\infty}(\breve{\Gamma}_{1})} \leq C e^{-c t},\\
&\|w_r\|_{L^{1} \cap L^{2} \cap L^{\infty}\left(\mathbb{R}\right)} \leq C t^{-\frac{3}{2}},\\
&\|w_r\|_{L^{1} \cap L^{2} \cap L^{\infty}(\partial D_{\varepsilon}(-\frac{\alpha}{6\beta}))} \leq C t^{-\frac{1}{3}},\\
&\|w_r\|_{L^{1} \cap L^{2} \cap L^{\infty}(\Sigma^{\varepsilon})} \leq C t^{-\frac{1}{3}}.
\end{aligned}
\end{equation}
\end{proposition}
\noindent The remainder of the proof to the asymptotic formula  proceeds as in sector $\mathcal{P}_{\leq}$.
 
 In  this work, we have investigated  the long time asymptotic solution of the Cauchy   problem for  defocusing Hirota equation with  decaying data in the special  transition region $|\frac{x}{t} -  \frac{\alpha^2}{3\beta}|t^{2/3}  \leq M$, $M$ is a positive constant.  Based on  the Riemann-Hilbert prblem which is established in \cite{Huang},  we perform the nonliear steepest descent method to analysis the asymptotic properties of the solution in the left  and right transition regions, respectively.  What is meaningful is that we find the solution of defocusing Hirota equation can be approximated  in terms of the real-valued solution of Painlev\'{e} $\mathrm{II}$ equation with the error $O (t^{-\frac{2}{3}})$.  What's more,  as for the focusing Hirota equation, we  can  derive  the asymptotic formula of the solution in transition regions  by nonlinear steepest descent method. It's worth noting that the result is similar to that of  defocusing case  except that other errors caused by soliton asymptotics.

\newpage

\noindent\textbf{Acknowledgements}

This work is supported by  the National Natural Science
Foundation of China (Grant No. 11671095,  51879045).

\end{document}